\documentclass{amsart}
\usepackage{amsfonts}
\usepackage{amsmath}
\usepackage{amssymb}
\usepackage{amsthm}
\usepackage{color}
\usepackage{comment}

\usepackage[foot]{amsaddr}

\usepackage{mathrsfs}

\usepackage{pdfpages}

\newtheorem{Thm}{Theorem}
\newtheorem{Lem}[Thm]{Lemma}
\newtheorem{Prop}[Thm]{Proposition}
\newtheorem{Cor}[Thm]{Corollary}

\theoremstyle{definition}
\newtheorem{Def}[Thm]{Definition}

\theoremstyle{remark}
\newtheorem{Rem}[Thm]{Remark}

\numberwithin{equation}{section}

\def\RR{{\mathbb{R}}}

\def\PP{{\mathbb{P}}}

\def\PP{{\mathbb{P}}}

\def\bigO{\mathcal{O}}

\title[Envy-free Cake cutting]{Envy-free Cake cutting:\\  a polynomial number of queries\\ with high probability}

\date{\today}

\author[G.~Ch\`eze]{Guillaume Ch\`eze}
\address{Guillaume Ch\`eze: Institut de Math\'ematiques de Toulouse\\
Universit\'e Paul Sabatier \\
118 route de Narbonne\\
31 062 TOULOUSE cedex 9, France}
\email{guillaume.cheze@math.univ-toulouse.fr}

\begin{document}


\begin{abstract}
In this article we propose a probabilistic framework in order to study the fair division of a divisible good, e.g. a cake, between $n$ players. Our framework contains two situations. The first corresponds to the ``Full independence model" used  in the study of fair division of indivisible goods. The second is in the spirit of smoothed analysis. We show that, in this framework,   there exists an envy-free division  algorithm satisfying  the following probability estimate:
$$\mathbb{P}\big( C(\mu_1, \ldots,\mu_n) \geq n^{7+b}\big) = \bigO\Big(n^{-\frac{b-1}{3}+1+o(1)}\Big),$$
where $\mu_1,\ldots, \mu_n$ correspond to the preferences of the $n$ players,
$C(\mu_1, \ldots,\mu_n)$ is the number of queries used by the algorithm and $b>4$.
In particular, this gives
$$\lim_{n \rightarrow + \infty}\mathbb{P}\big( C(\mu_1, \ldots,\mu_n) \geq n^{12}\big) = 0.$$
It must be noticed that nowadays few things are known about the complexity of envy-free division algorithms. Indeed, Procaccia has given a lower bound in $\Omega(n^2)$ and Aziz and Mackenzie have given an upper bound in $n^{n^{n^{n^{n^{n}}}}}$. As our estimate means that we have $C(\mu_1, \ldots, \mu_n)<n^{12}$ with a high probability, this gives a new insight on the complexity of envy-free cake cutting algorithms.\\
Our result follows from a study of Webb's algorithm and a theorem of Tao and Vu about the smallest singular value of a random matrix.\\

\end{abstract}


\maketitle

\section*{Introduction}
In this article we study the problem of fair resource allocation.  The goal in this problem is to share a  heterogeneous good between different players or agents. This good can be for example: a cake, land, time or computer memory. This problem is old. For example, the Rhind mathematical papyrus contains problems about the division of loaves of bread and about the partition of plots of land. In the Bible we find the famous ``Cut and Choose" algorithm between Abraham and Lot, and in the greek mythology we find the trick at Mecone. More recently, the ``Cut and Choose" protocol has been used in the United Nations Convention on the Law of the Sea (December 1982, Annex III, article 8).\\

  The problem of fair division has been formulated in a scientific way by Steinhaus in 1948, see \cite{Steinhaus}. Nowadays, there exist several papers, see e.g. \cite{AzizMackenzie, BJK, BramsTaylorarticle, DubinsSpanier, EdmondsPruhs, EvenPaz, Pikhurko, Procacciasurvey, RoberstonWebbarticle, Thomson2006}, and books about this topic, see e.g. \cite{Barbanel, BramsTaylor, Procacciachapter, RobertsonWebb}. These results appear in the mathematics, economics, political science, artificial intelligence and computer science literature. Recently, the cake cutting problem has been studied intensively by computer scientists for solving resource allocation problems in multi agents systems, see e.g.~\cite{Branzei, Chen, Chevaleyre06, Dynamic,Sziklai}. \\

 Throughout this article, the cake will be a heterogeneous good represented by the interval $\mathcal{C}=[0,1]$. This assumption is classical and not restrictive for our study.\\
We also consider $n$ players and we associate to each player a non-atomic probability measure $\mu_i$ on the interval $\mathcal{C}=[0,1]$. More precisely, we suppose that the measures $\mu_i$ are  absolutely continuous with respect to the Lebesgue measure.
These measures represent the preferences, the utility functions of the players. We have $\mu_i(\mathcal{C})=1$ for all $i$.\\
 The problem in this situation is to get a fair division of $\mathcal{C}=\mathcal{C}_1\sqcup \ldots \sqcup \mathcal{C}_n$, where the $i$-th player gets $\mathcal{C}_i$.\\

When we study fair divisions, we have to define   ``fair" precisely. Indeed, several notions exist.
\begin{enumerate}
\item[$\bullet$] We say that a division is \emph{proportional} when for all $i$, we have 
$$\mu_i(\mathcal{C}_i) \geq 1/n.$$
\item[$\bullet$] We say that a division is \emph{equitable} when for all $i\neq j$, we have 
$$\mu_i(\mathcal{C}_i)=\mu_j(\mathcal{C}_j).$$
\item[$\bullet$] We say that a division is \emph{exact} in the ratios $(\alpha_1,\alpha_2,\ldots,\alpha_n)$, where $\alpha_i \geq 0$ and $\alpha_1+\alpha_2+\cdots+\alpha_n=1$, when for all $i$ and $j$ we have 
$$\mu_i(\mathcal{C}_j)=\alpha_j.$$
\item[$\bullet$] We say that a division is \emph{envy-free} when for all $i \neq j$, we have 
$$\mu_i(\mathcal{C}_i) \geq  \mu_i(\mathcal{C}_j).$$
\end{enumerate}
There also exist several classical properties to study a fair division: Pareto optimality, monotonicity and contiguity, see e.g.~\cite{AzizYe,Segal-Halevi}.\\
 Our paper  deals with a practical problem: the computation  of envy-free fair divisions.\\

 In order to describe algorithms we  need a model of computation. There exist two main classes of cake cutting algorithms: discrete and continuous protocols (also called moving knife methods). Here, we study discrete algorithms. These kinds of algorithms can be  described thanks to the  classical model introduced by Robertson and Webb and formalized by Woeginger and Sgall in \cite{Woeg}. In this model we suppose that a mediator interacts with the agents. The mediator asks two type of queries: either cutting a piece with a given value, or evaluating a given piece. More precisely, the two type of queries allowed are:
\begin{enumerate}
\item $eval_i(x,y)$: Ask agent $i$ to evaluate the interval $[x,y]$. This means return $\mu_i([x,y])$.
\item $cut_i(x,a)$: Ask agent $i$ to cut a piece of cake $[x,y]$ such that $\mu_i([x,y])=a$. This means: for given $x$ and $a$, return $y$ such that $\mu_i([x,y])=a$.
\end{enumerate} 
In the Robertson-Webb model the mediator can adapt the queries based on  the previous answers given by the players. In this model, the complexity counts the number of queries necessary to get a fair division. For a rigorous description of this model we can consult: \cite{Branzei2017,Woeg}.\\
We can remark that this model of computation does not take into account the nature and the number of operations performed by the mediator. The Blum-Shub-Smale-Robertson-Webb model of computation introduced in \cite{ChezeBSSRW}  avoids these drawbacks.\\

The first studied notion of fair division has been proportional  division, \cite{Steinhaus}. Proportional  division is a simple and well understood notion. In \cite{Steinhaus} Steinhaus explains the Banach-Knaster algorithm which gives a proportional  division. There also exists an optimal algorithm to compute a proportional  division in the Robertson-Webb model, see \cite{EdmondsPruhs, EvenPaz}. The complexity of this algorithm is in $\bigO\big(n\log(n)\big)$. Furthermore, the portion $\mathcal{C}_i$ given to the $i$-th player in this algorithm is an interval.\\

Exact divisions in the ratios $(\alpha_1, \ldots, \alpha_n)$ exist for all ratios $(\alpha_1, \ldots, \alpha_n)$. The existence of this kind of fair division follows from a convexity theorem given by Lyapounov, see e.g. \cite{DubinsSpanier}. When we have  $\alpha_i=1/n$, for all $i$, we just say that the division is exact. Unfortunately, there exists no algorithm to compute exact divisions, see \cite{RobertsonWebb}.\\
Equitable fair division is of the same kind.  Indeed, there exist equitable fair divisions where each $\mathcal{C}_i$ is an interval, see \cite{Cechexistence,Chezeequitable,Segal-Halevi}. However, there do not exist discrete protocols computing an equitable fair division, see \cite{Cech,ChezeBSSRW,ProcWang}. \\

Envy-free fair division is difficult to obtain in practice. Indeed, whereas envy-free fair divisions where each $\mathcal{C}_i$ is an interval exist, there does not exist an algorithm in the Robertson-Webb model computing such divisions. These results have been proved by Stromquist in \cite{Stromquistexist,Stromquist}.\\
 The first envy-free algorithm for $n$ players has been given by Brams and Taylor in  \cite{BramsTaylorarticle}. This algorithm was discovered approximatively 50 years after the first algorithm computing a proportional fair division. The Brams-Taylor algorithm has an unbounded complexity in the Robertson-Webb model. This means that we cannot bound the complexity of this algorithm in terms of the number of players only. A complexity study with ordinal numbers of these algorithms has been done in \cite{Gasarch}. It is only recently that a finite and bounded algorithm has been given to solve this problem, see \cite{AzizMackenzie}. The complexity of this algorithm is in $\bigO\Big(n^{n{^{n^{n^{n^{n}}}}}}\Big)$. When $n=2$, $n^{n{^{n^{n^{n^{n}}}}}}$ is bigger than the number of atoms in the universe\ldots
 A lower bound for envy-free division algorithm has been given by Procaccia in \cite{Procaccia-lowerbound}. This lower bound is in $\Omega(n^2)$.\\

We can remark that there is a huge difference between the complexity in the worst case $\bigO\Big(n^{n{^{n^{n^{n^{n}}}}}}\Big)$ and the lower bound  $\Omega(n^2)$. Therefore a natural question  arises: \\

\emph{Can we design an envy-free algorithm such that in practice the number of queries is smaller than $n^d$, where $d$ is a given degree, with a high probability ?}\\

In order to answer this question we have to define a probabilistc framework.\\
When we consider indivisible goods there exist two probabilistic models, see e.g.~ \cite{Bouveret2011}.\\

 The first model is the \emph{Full correlation model}. In this model we suppose that all agents have the same preference and all preferences are equiprobable. When we want to share several indivisible goods, this case corresponds to the worst case. However, in the cake cutting situation, if we suppose that $\mu_1=\ldots=\mu_n$ then we can easily obtain an envy-free division. Indeed, we ask the first player to cut the cake in $n$ equals portions. Thus the full correlation model is  not an interesting model in the cake cutting situation.\\
 
 The second model in the indivisible goods setting is the \emph{Full independence model}. In this model we suppose that all preferences are equiprobable and that all agents have independent preferences.  In the cake cutting setting, in order to obtain a similar situation we consider the following construction:\\
 
 First, we divide the interval $[0,1]$ into $n$ witness intervals 
 $$W_j=\Big[\dfrac{j-1}{n},\dfrac{j}{n}\Big], \textrm{ where  }j=1, \ldots,n.$$ 
Second, we remark that for all probabilistic measures $\mu_i$ on $[0,1]$, the vector\\ $\big(\mu_i(W_1),\ldots,\mu_i(W_n)\big)$ belongs to the standard $(n-1)$-simplex. Indeed, as $\mu_i$ is a probabilistic measure we have for all $j=1, \ldots,n$, $\mu_i(W_j) \geq 0$,  and 
$$\mu_i(W_1)+\cdots+\mu_i(W_n)=\mu_i(W_1 \sqcup \ldots \sqcup W_n)=\mu_i([0,1])=1.$$
When we consider a random measure $\mu_i$, it is natural to suppose that all witness intervals play the same role. For example, there is no reason to suppose that the players usually prefer the first part $W_1$ of the cake.\\

Our first probabilistic situation is thus the following:\\

 \emph{We suppose that the distribution of $\big(\mu_i(W_1),\ldots,\mu_i(W_n)\big)$ follows a uniform distribution over the standard $(n-1)$-simplex.}\\

A classical  way to obtain a uniform distribution on the standard $(n-1)$-simplex is the following, see \cite[Theorem 4.1]{Devroye}:
Consider $n$ independent random variables $X_i$ with probability density function $f_i(x)=e^{-x}$. Set $S=\sum_{i=1}^n X_i$ and $Y_i=X_i/S$, then $(Y_1,\ldots,Y_n)$ follows the uniform distribution on the standard $(n-1)$-simplex.\\

Furthermore, \emph{in this first setting we are going to suppose that agents have independent preferences}. This means for example that $\mu_1(W_j)$ is independent of $\mu_2(W_j)$. \\

Thus, in the cake cutting situation the \emph{Full independence model} means that we suppose that the following hypothesis holds:\\

\emph{ $H_1$: Taking randomly a matrix $\mathcal{M}=\big(\mu_i(W_j)\big)$
 means that we consider a random matrix $\mathcal{M}=(m_{ij})$ where 
$$m_{ij}=\dfrac{X_{ij}}{\sum_{k=1}^n X_{ik} }$$ and $X_{ij}$ are independent exponential random variables, i.e. with a probability density function $f_{ij}(x)=e^{-x}$ defined over $[0,+\infty[$.}\\

Our second probabilistic situation is in the spirit of smoothed analysis. Indeed, we are going to study the worst case complexity of an algorithm under slight random perturbations. More precisely, we consider this kind of situation:\\
There is a partition $[0,1]=\sqcup_{i=1}^n W_j$, but $W_j$ are not necessarily equal to $\Big[\dfrac{j-1}{n},\dfrac{j}{n}\Big]$.\\
Each agent evaluate all $W_j$ and then send these informations to the mediator. However, these values are transmitted over a communication channel in the presence of noise.\\

Therefore, the idea is  to consider a matrix $\mathcal{M}=(m_{ij})$ where 
\begin{itemize}
\item $m_{ij}=a_{ij}+\varepsilon_{ij}$,
\item $a_{ij} \in \RR^+$ correspond to $\mu_i(W_j)$ without perturbation, thus $\sum_{j=1}^n a_{ij}=1$, 
\item $\varepsilon_{ij}$ are independent and identically distributed random variables  with mean zero. They  correspond to  perturbations.
\end{itemize}

However, we cannot consider directly these kinds of matrices in our study for two reasons.\\
First, if the pertubation $\varepsilon_{ij}$ satisfies $\varepsilon_{ij}<-a_{ij}$  then we  get $m_{ij}=a_{ij}+\varepsilon_{ij}<0$. Thus we cannot interpreted $m_{ij}$ as a measure of $W_j$. Therefore, we are going to suppose that there exists $\varepsilon>0$ such that $a_{i,j}>\varepsilon$, for all $i,j$, and $|\varepsilon_{ij}|<\varepsilon$.
With this assumption we have $a_{ij}+\varepsilon_{ij}>0$. We remark that this means that all players are hungry: each interval $W_j$ have a positive value for all players. \\
Second, even if $\sum_{j=1}^n a_{ij}=1$, we do not necessarily have $\sum_{j=1}^n a_{ij}+\varepsilon_{ij}=1$. Then, the matrix $\mathcal{M}$ is not stochastic. This problem can be avoided if we consider the matrix with coefficients $m_{ij}=\dfrac{a_{ij}+\varepsilon_{ij}}{\sum_{k=1}^n (a_{ik}+\varepsilon_{ik}) }$.\\

Thus, in this  second situation, we consider the following hypothesis parametrized by $\varepsilon \in \,]0,1[$.\\

\emph{ $H_2(\varepsilon)$: Taking randomly a matrix $\mathcal{M}=\big(\mu_i(W_j)\big)$
 means that we consider a  matrix \mbox{$\mathcal{M}=(m_{ij})$} where 
$$m_{ij}=\dfrac{a_{ij}+\varepsilon_{ij}}{\sum_{k=1}^n (a_{ik}+\varepsilon_{ik}) }$$
$$ a_{ij}> \varepsilon , \quad \sum_{j=1}^na_{ij}=1.$$
and $\varepsilon_{ij}$ are independent and identically distributed random variables  with mean zero  such that $-\varepsilon < \varepsilon_{ij} <\varepsilon$.\\
}

\emph{In this second setting, we do not suppose that agents have independent preferences, we suppose that the noise correspond to independent and identically distributed random variables}.\\


We remark that in the two situations we have defined what is a random matrix $\mathcal{M}=\big(\mu_i(W_j)\big)$ and that we do not have defined what is a random measure. Indeed, instead of taking random measures and then constructing the matrix $\mathcal{M}$, we have directly defined the probability distribution for the matrix $\mathcal{M}$. This approach allows to obtain a simple and explicit probabilistic framework.\\

In this article, we  study the complexity of an envy-free fair division algorithm. We will denote by $C(\mu_1, \ldots,\mu_n)$ the number of queries used by this algorithm when the inputs are $\mu_1$, $\mu_2$, \ldots, $\mu_n$.

\begin{Thm} \label{thm:main} If we suppose that the hypothesis $H_1$ or $H_2(\varepsilon)$ is satisfied, then we have the following result:\\
There exists a protocol in the Robertson-Webb model of computation giving an envy-free fair division and such that for all $b>4$ we have the following probability estimate
$$\mathbb{P}\big( C(\mu_1, \ldots,\mu_n) \geq n^{7+b}\big) = \bigO\Big(n^{-\frac{b-1}{3}+1+o(1)}\Big).$$
\end{Thm}

This theorem says: the bigger the number of queries, the smaller the probability.\\

We recall that $f(n)=\bigO\big(g(n)\big)$ means that  there exists a constant $C$ and an integer $n_0$ such that for all $n\geq n_0$, we have $|f(n)|\leq Cg(n)$.\\
The notation $o(1)$ refers to a function $f(n)$ such that $\lim_{n \rightarrow +\infty}  f(n)=0$.\\

\textbf{Examples:}\\
$\bullet$ If we choose $b=5$ then
$$-\dfrac{b-1}{3}+1+o(1) =-\dfrac{1}{3}+o(1).$$
When $n$ is big enough we can suppose $o(1)<1/6$ and then in this case
 $$\bigO\Big(n^{-\frac{b-1}{3}+1+o(1)}\Big) = \bigO\big(n^{-1/6}\big).$$
 This gives
 $$\lim_{n \rightarrow +\infty}  \mathbb{P}\big( C(\mu_1, \ldots,\mu_n) \geq n^{12}\big) = 0.$$\\
 
\noindent$\bullet$ If we choose $b=11$ then 
$$-\dfrac{b-1}{3}+1+o(1) =-\dfrac{7}{3}+o(1).$$
 When $n$ is big enough we can suppose $o(1)<1/3$ and then in this case
 $$\bigO\Big(n^{-\frac{b-1}{3}+1+o(1)}\Big) = \bigO\big(n^{-2}\big).$$
Thus Theorem~\ref{thm:main} gives
$$\mathbb{P}\big( C(\mu_1, \ldots,\mu_n)\geq n^{18} \big)= \bigO\Big( \dfrac{1}{n^{2}}\Big).$$

These bounds are not very sharp but they give a precise statement of the following idea:  when $n$ is big the probability that the algorithm uses more than $n^{12}$ (or $n^{18}$) queries  is very small.\\
At last, we remark that this theorem does not give the expected numbers of queries used by the algorithm.\\

\textbf{Strategy of the algorithm and structure of the paper}\\
The algorithm proposed in this article is just a slight modification of Webb's super envy-free division algorithm. Webb's algorithm constructs an envy-free division from the matrix $\mathcal{M}=\big(\mu_i(W_j)\big)$ when $\det(\mathcal{M})\neq 0$. The algorithm that we propose works as follows: if $\det(\mathcal{M})\neq 0$ then use Webb's algorithm else use another envy-free algorithm.\\
When the hypothesis $H_1$ or $H_2(\varepsilon)$ is satisfied the probability that $\det(\mathcal{M})=0$ is equal to zero. Thus, in practice our algorithm almost always corresponds to Webb's algorithm. As the number of queries needed in Webb's algorithm can be written in terms of the smallest singular value of $\mathcal{M}$, the strategy to prove Theorem~\ref{thm:main} relies on a probabilistic study of the smallest singular value of $\mathcal{M}$.\\

The structure of this article in thus the following:\\
In the first section, we recall what is a super envy-free fair division and we also recall Webb's super envy-free algorithm. Then, we give our algorithm. In Section~2, we study the number of queries used by this algorithm. This leads us to recall some standard results on singular values of a matrix and to write the complexity of the algorithm in terms of the smallest singular value of the matrix $\mathcal{M}$. In Section 3, we use a theorem of Tao and Vu, see \cite{TaoVu}, about the probability that $\mathcal{M}$ has a small singular value. This theorem will be the key point in the proof of Theorem~\ref{thm:main}.

\section{The algorithm}
\subsection{Super envy-free algorithm}
Super envy-free fair division is a strong kind of envy-free division. This notion has been introduced and studied by Barbanel, see \cite{Barbanelarticle,Barbanel}.
\begin{Def}
We say that a division is \emph{super envy-free} when for all $i \neq j$, we have  
$$\mu_i(\mathcal{C}_i)>\dfrac{1}{n}>\mu_i(\mathcal{C}_j).$$
\end{Def}

This definition says that this division is proportional and all players think to have stricly more than  other players.  Of course, this kind of fair division is not always possible. For example, if $\mu_1=\mu_2=\cdots=\mu_n$, then the previous inequality is not possible. Indeed,  we cannot have $\mu_1(\mathcal{C}_1)>1/n>\mu_2(\mathcal{C}_1)=\mu_1(\mathcal{C}_1)$.\\
However, a super envy-free fair division exists when the measures $\mu_i$ are linearly independent.
\begin{Def}
Let $\mu_1$, \ldots, $\mu_n$ be $n$ measures on a measurable set $(\mathcal{C},\mathcal{B})$, where $\mathcal{B}$ is the Borel $\sigma$-algebra.  We say that these measures are linearly independent when they are linearly independent as functions from $\mathcal{B}$ to $[0,1]$.
\end{Def}

\begin{Thm}[Barbanel's theorem]
A super envy-free division exists if and only if the measures $\mu_1$, \ldots, $\mu_n$ are linearly independent.
\end{Thm}

In the following we are going to use a witness matrix in order to know if the measures are linearly independent.

\begin{Def}
The witness matrix associated to the partition $\mathcal{C}=W_1\sqcup \ldots \sqcup W_n$ and the measures $\mu_1$, \ldots, $\mu_n$ is the matrix
 $\mathcal{M}=\big(\mu_i(W_j)\big)$.
\end{Def}

\begin{Rem}
If $\det(\mathcal{M})\neq 0$ then the measures $\mu_1$, \ldots, $\mu_n$ are linearly independent.
\end{Rem}

In \cite{Webbalgo}, Webb gives a strategy to compute a super envy-free fair division. In order to recall this strategy, we recall what is an $\varepsilon$ \emph{near-exact fair division}.\\

\begin{Def}
Let $A$ be a measurable subset of $\mathcal{C}$.\\
We say that a division of $A=A_1\sqcup \ldots \sqcup A_n$ is $\varepsilon$ \emph{near-exact} in the ratios $(\alpha_1,\alpha_2,\ldots,\alpha_n)$, where $\alpha_i \geq 0$ and $\alpha_1+\alpha_2+\cdots+\alpha_n=1$, when for all $i$ and $j$ we have
$$|\mu_i(A_j)-\alpha_j \times \mu_i(A) |<\varepsilon \times  \mu_i(A).$$
\end{Def}

Now, we can describe Webb's algorithm.\\

\texttt{Super Envy-free fair division algorithm}\\
\texttt{Inputs:} A partition $\mathcal{C}=W_1 \sqcup \ldots \sqcup W_n$, a matrix $\mathcal{M}_0=(m_{ij})$ where $m_{ij}=\mu_i(W_j)$, $\mathcal{M}_0$ is non-singular.\\
\texttt{Outputs:} A super envy free division $\mathcal{C}=\mathcal{C}_1 \sqcup \ldots \sqcup \mathcal{C}_n$.
\begin{enumerate}
\item Compute $\mathcal{M}_0^{-1}=(\tilde{m}_{ij})$.
\item \label{step2} Set $\delta:=\dfrac{n-1}{n(1-tn)}$ where $t=\min_{i,j}(\tilde{m}_{ij})$.
\item Set $\mathcal{N}:=(n_{ij})$, where $n_{ii}:= 1/n  +\delta$ and $n_{ij}=1/n -\delta/(n-1)$.
\item Compute $\mathcal{R}=(r_{ij}):=\mathcal{M}_0^{-1}\mathcal{N}$.
\item\label{step5} For $j=1, \ldots, n$ do\\
\hphantom{bla} Compute an $\varepsilon=\delta/{n^2}$ near-exact fair division of $W_j$ in the ratios $(r_{j1},\ldots, r_{jn})$,\\
\hphantom{bla} this gives $W_j=W_{j1} \sqcup \ldots \sqcup W_{jn}$.
\item For all $i=1, \ldots, n$ do\\
\hphantom{bla} $\mathcal{C}_i:= W_{1i} \sqcup W_{2i} \sqcup \ldots \sqcup W_{ni}.$\\
\end{enumerate}

We remark that in Step~\ref{step2}, we have  $t\leq 0$. Indeed, if $t>0$ then the equality $\mathcal{M}_0\mathcal{M}_0^{-1}=I$ is impossible, because the coefficients $m_{ij}$ are non-negative. Therefore in Step~\ref{step2}, we have  $\delta>0$. The formula used to define $\delta$ is constructed in such a way that the coefficients $r_{ij}$ of $\mathcal{R}$ are non-negative.\\
Furthermore, we have $\sum_{j=1}^n r_{ij}=1$. Indeed, if we set $e=\,^t(1,\ldots,1)$ then as $\mathcal{M}_0$ is a stochastic matrix, we have $\mathcal{M}_0\cdot e=e$ and  $\mathcal{M}_0^{-1}\cdot e=e$. Moreover, by construction $\mathcal{N}$ is also a stochastic matrix, then $\mathcal{N}\cdot e=e$. It follows $\mathcal{R}\cdot e= \mathcal{M}_0^{-1}\mathcal{N}\cdot e=e$, which gives the equality $\sum_{j=1}^n r_{ij}=1$.\\

In order to explain this algorithm, suppose that in Step~\ref{step5} we compute an exact fair division in the ratios $(r_{j1},\ldots,r_{jn})$ instead of an $\varepsilon$ near- exact fair division with these ratios. Then, by construction the partition $\mathcal{C}=\mathcal{C}_1\sqcup \ldots \sqcup \mathcal{C}_n$ has the following property: $\mu_i(\mathcal{C}_j)=n_{ij}$. This gives $\mu_i(\mathcal{C}_i) =1/n +\delta$ and $\mu_i(\mathcal{C}_j)=1/n-\delta/(n-1)$. Thus this partition gives a super envy-free division.\\
In practice, the computation of an exact fair division in the ratios $(r_{j1},\ldots,r_{jn})$ is impossible, since it has been proved that such algorithms cannot exist, see \cite{RobertsonWebb}. That is the reason why an $\varepsilon$ near-exact algorithm is used. Indeed, an $\varepsilon$ near-exact algorithm in the Robertson-Webb model exists, see \cite[Theorem 10.2]{RobertsonWebb}. Therefore the  idea is to choose a small enough $\varepsilon$ in order to obtain  a result very close to the theoretical result where $\varepsilon=0$. Thus, we obtain in practice a partition where $\mu_i(\mathcal{C}_j)$ are very close to $n_{ij}$ and then the division is super envy-free.\\

The number of queries used by the $\varepsilon$ near-exact division algorithm is at most $n\times (2+2n^{3/2})/\varepsilon$, see \cite[Theorem 10.2]{RobertsonWebb}. As already remarked in \cite{RobertsonWebb}, this bound can be improved. However,  we just want to get a bound on $C(\mu_1,\ldots, \mu_n)$ in terms of a polynomial in $n$, thus  this estimate is sufficient.\\
 Therefore, the number of queries used by Webb's super envy-free algorithm is at most $n^4\times (2+2n^{3/2})/\delta$, since $\varepsilon=\delta/n^2$ and in Step~\ref{step5} we compute $n$ $\varepsilon$-near-exact fair divisions. Thanks to the definition of $\delta$  we get the following lemma.

\begin{Lem}\label{lem:boundsuperalgo}
The number of queries used by the super envy-free division algorithm in the Robertson-Webb model is bounded by
$$  \dfrac{n^5\times (2+2n^{3/2} ) \times (1-tn)}{n-1} \in \bigO\Big(\max\big(1,|t|\big)n^{6.5}\Big).$$
\end{Lem}


\subsection{An envy-free algorithm}$\,$\\

\texttt{Envy-free fair division algorithm}\\
\texttt{Inputs:} A partition $\mathcal{C}=W_1 \sqcup \ldots \sqcup W_n$, $n$ measures $\mu_1$, \ldots, $\mu_n$.\\
\texttt{Outputs:} An envy free division $\mathcal{C}=\mathcal{C}_1 \sqcup \ldots \sqcup \mathcal{C}_n$.
\begin{enumerate}
\item \% \emph{ Construct the matrix $\mathcal{M}=(m_{ij})$ where $m_{ij}=\mu_i(W_j)$.} \%\\
For all $i=1, \ldots,n$, do \\
\hphantom{bla} For all $j=1, \ldots, n$ do\\
\hphantom{bla} $m_{ij}:=eval_i(W_j)$.
\item If $\det(\mathcal{M})=0$ then compute an envy-free fair division  thanks to  Aziz-Mackenzie's algorithm,\\
Else compute a super envy-free fair division  thanks to Webb's algorithm.
\end{enumerate}

\begin{Rem}
When $\det(\mathcal{M}) =0$, we have to use an algorithm different from Webb's algorithm. Indeed, in this case Webb's algorithm is not defined (we cannot compute $\mathcal{M}^{-1}$). Furthermore, it is not necessary to use the  Aziz-Mackenzie's algorithm. The bound given in Theorem~\ref{thm:main} will not change if we use another envy-free algorithm when $\det(\mathcal{M}) =0$.\\
\end{Rem}

\begin{Rem}
If $\det(\mathcal{M}) =0$ then we can try another partition $\mathcal{C}=W'_1 \sqcup \ldots \sqcup W'_n$. However, if the measures are linearly dependent then for all partitions we have $\det(\mathcal{M})=0$.
\end{Rem}
\section{Complexity analysis}
The number of queries used by our envy-free division algorithm depends on $|t|$ when $\det(\mathcal{M}) \neq 0$. In the next subsection, we are going to bound $|t|$ by $\sigma_n^{-1}$ where $\sigma_n$ is the smallest singular value  of $\mathcal{M}$. Then, in the second subsection, we use an estimate on the probability $\PP(\sigma_n \leq n^{-b})$ in order to prove our theorem.
\subsection{The smallest singular value}
We recall here the definition and a simple result about the singular values of a matrix.

\begin{Def}
The singular values of a matrix $\mathcal{M}$ are the square roots of the eigenvalues of $\mathcal{M}^T\mathcal{M}$.  They are denoted by $\sigma_1(\mathcal{M}) \geq  \cdots \geq \sigma_n(\mathcal{M})$.
\end{Def}

\begin{Rem}
We have $\det(\mathcal{M})=0 \iff \sigma_n(\mathcal{M})=0$.
\end{Rem}

The smallest singular value allows to bound the coefficients of the inverse of a matrix.

\begin{Prop}\label{prop:sing}
Let $\mathcal{M}$ be a non-singular matrix, such that $\mathcal{M}^{-1}=(\tilde{m}_{ij})$. \\
Let $\sigma_n(\mathcal{M})$ be the smallest singular value of $\mathcal{M}$.
We have
$$\max_{ij}\big( |\tilde{m}_{ij}| \big) \leq \|\mathcal{M}^{-1}\|_2 =\sigma_n^{-1}(\mathcal{M}).$$
\end{Prop}

\begin{proof}
This is a classical result, see \cite[Formula 2.3.8 page 56]{GolubvanLoan} and \cite[Theorem~3.3]{Demmel}.
\end{proof}

The previous proposition allows us to obtain an upper bound on the complexity of the super envy-free algorithm in terms of $\sigma_n(\mathcal{M})$.\\

\begin{Cor}\label{cor:boundsuper}
When $n \geq 19$, the number of queries used by the super envy-free division algorithm in the Robertson-Webb model is bounded by 
$$n^7 \times  \max\big(1, \sigma_n^{-1}(\mathcal{M})\big).$$
\end{Cor}

\begin{proof}
We have already remarked that in Step~\ref{step2} of the super envy-free algorithm we have $t=\min_{i,j} (\tilde{m}_{ij}) \leq 0$. Then Proposition~\ref{prop:sing} gives 
$$-t=|t|\leq \max_{i,j}|\tilde{m}_{ij}| \leq \sigma_n^{-1}(\mathcal{M}).$$
 Then Lemma~\ref{lem:boundsuperalgo}    implies that the number of queries used by the super envy-free division algorithm in the Robertson-Webb model is bounded by
$$   \dfrac{n^5\times (2+2n^{3/2} ) \times (1+n\sigma_n^{-1}(\mathcal{M}))}{n-1}.$$

As we have supposed that $n \geq 19$, we get
$$\dfrac{2+2n^{3/2}}{n-1} \leq \dfrac{n}{2}.$$
Then, we have 
\begin{eqnarray*}
\dfrac{n^5\times (2+2n^{3/2} ) \times (1+n\sigma_n^{-1}(\mathcal{M}))}{n-1} &\leq & \dfrac{n^6}{2}\times (1+n\sigma_n^{-1}(\mathcal{M}))\\
& \leq &
\dfrac{n^7}{2}\big(1+\sigma_n^{-1}(\mathcal{M})\big)\\
& \leq &
n^7 \times  \max\big(1, \sigma_n^{-1}(\mathcal{M})\big).
\end{eqnarray*}
\end{proof}

The singular value $\sigma_n(\mathcal{M})$ measures how far $\mathcal{M}$ is  from a singular matrix. Therefore, if $\sigma_n(\mathcal{M})$ is small then the measures $\mu_i$ are nearly linearly dependent and the previous corollary shows that the number of queries is big. This result satisfies the general result: if the agents have ``very different" preferences it will be easier to get an envy-free fair division.
A precise statement of this result with an explicit bound  has been  given in 
\cite[Corollary 17]{ChezeLAA}.


\subsection{Proof of Theorem~\ref{thm:main}}$\,$\\
We consider the event $C(\mu_1, \ldots,\mu_n) \geq n^{7+b}$. \\
During the envy-free algorithm two situations appear:\\

First, $\det(\mathcal{M}) \neq 0$, then in this situation the algorithm uses the super envy-free algorithm. Thanks to Corollary~\ref{cor:boundsuper}, the number of queries used  in this situation satisfies
$$n^{7+b} \leq C(\mu_1, \ldots,\mu_n) \leq n^7 \times  \max\big(1, \sigma_n^{-1}(\mathcal{M})\big).$$
As $n > 1$, it follows
$$n^{b} \leq \sigma_n^{-1}(\mathcal{M}).$$
This means that we have the following inclusion
$$(\star) \quad \{\det(\mathcal{M})\neq 0\} \cap \{ C(\mu_1, \ldots,\mu_n) \geq n^{7+b} \} \subset  \{ \sigma_n(\mathcal{M})\leq n^{-b} \}.$$

The second situation corresponds to $\det(\mathcal{M})=0$. Thus the second situation corresponds to $\sigma_n(\mathcal{M})=0$ and obviously $\sigma_n(\mathcal{M})\leq n^{-b}$. This gives the following inclusion
$$ (\star \star) \quad \{\det(\mathcal{M}) =0 \} \subset  \{ \sigma_n(\mathcal{M})\leq n^{-b} \}.$$

Thanks to $(\star)$ and $(\star \star)$ we get
$$\{ C(\mu_1, \ldots,\mu_n) \geq n^{7+b} \} \subset  \{ \sigma_n(\mathcal{M})\leq n^{-b} \}.$$
We deduce then the following inequality between probabilities
$$\PP\big( C(\mu_1, \ldots,\mu_n) \geq n^{7+b} \big) \leq \PP\big(  \sigma_n(\mathcal{M})\leq n^{-b} \big).$$

In the next section we are going to prove the following proposition.
\begin{Prop}\label{prop:tao}
If we  suppose that the hypothesis $H_1$  or $H_2(\varepsilon)$ is satisfied then the following holds:\\
Let $b>4$ be a constant, then there exists a constant $c>0$ depending on $b$ such that
$$\PP\big(\sigma_n(\mathcal{M}) \leq n^{-b}\big)\leq c\Big( n^{-\frac{b-1}{3}+1+o(1)}+ne^{-\sqrt{n}}+n^2e^{-\sqrt{n}-1}\Big).$$
\end{Prop}

In order to finish the proof of Theorem~\ref{thm:main},  we remark that 
for all $b>4$ we have $ne^{-\sqrt{n}}+n^2e^{-\sqrt{n}-1}= \bigO( n^{-\frac{b-1}{3}+1})$.
Therefore, we get
$$\PP\big( C(\mu_1, \ldots,\mu_n) \geq n^{7+b} \big) \leq \PP\big(\sigma_n(\mathcal{M}) \leq n^{-b}\big)= \bigO \Big(  n^{-\frac{b-1}{3}+1+o(1)}\Big),$$
which gives the desired estimate.

\section{An estimate for $\PP(\sigma_n(\mathcal{M}) \leq n^{-b})$}
In this section we prove Proposition~\ref{prop:tao}.\\
In order to bound $\PP\big(\sigma_n(\mathcal{M}) \leq n^{-b}\big)$, we are going to introduce some notations.\\
As we suppose that the hypothesis $H_1$ or $H_2(\varepsilon)$ is satisfied we have
$$\mathcal{M}=\mathcal{D}\mathcal{X},$$
where $\mathcal{D}$ is a diagonal matrix. Indeed:\\
\begin{itemize}
\item If $H_1$ is satisfied, $\mathcal{D}$ is the diagonal matrix with coefficient in the $i$-th row equal to $1/(\sum_{j=1}^n X_{ij})$, $\mathcal{X}$ is the matrix with coefficients $X_{ij}$ and $X_{ij}$ are independent exponential random variables with probability density function $f_{ij}(x)=e^{-x}$.
\item If $H_2(\varepsilon)$ is satisfied, $\mathcal{D}$ is the diagonal matrix with coefficient in the $i$-th row equal to $1/(\sum_{j=1}^n a_{ij}+\varepsilon_{ij})$, $\mathcal{X}$ is the matrix with coefficients $a_{ij}+\varepsilon_{ij}$ and $\varepsilon_{ij}$ are independent and identically distributed random variables with mean zero and such that $-\varepsilon< \varepsilon_{ij} < \varepsilon$.
\end{itemize}

We also consider the two following events:
$$A=\{\sigma_n(\mathcal{M}) \leq n^{-b}\} \cap \{\sigma_n(\mathcal{D})\leq n^{-3/2}\}.$$
$$B= \{\sigma_n(\mathcal{M}) \leq n^{-b}\} \cap \{\sigma_n(\mathcal{D})\geq n^{-3/2}\}.$$
Obviously, we have 
$$(\sharp) \quad \{\sigma_n(\mathcal{M}) \leq n^{-b} \} =A \cup B.$$
$\,$\\

Now, we are going to bound $\PP(A)$ and $\PP(B)$.

\begin{Lem}\label{lem:p(A)}
If $H_1$ is satisfied then we have 
$$\PP(A) \leq ne^{-\sqrt{n} }.$$
\end{Lem}

\begin{proof}
We have $A \subset \{ \sigma_n(\mathcal{D}) \leq n^{-3/2} \}$. Thus
$$\PP(A) \leq \PP\big( \sigma_n(\mathcal{D}) \leq n^{-3/2}\big).$$
As $\mathcal{D}$ is a diagonal matrix with coefficients $1/(\sum_{j=1}^n X_{ij})$, we have
\begin{eqnarray*}
\sigma_n(\mathcal{D}) \leq n^{-3/2} & \Rightarrow & \min_i \Big(\dfrac{1}{ X_{i1}+\cdots+X_{in} } \Big) \leq n^{-3/2}\\
& \Rightarrow & \exists i_0, \, X_{i_01}+\cdots+X_{i_0n} \geq n^{3/2}\\
&\Rightarrow & \exists (i_0,\,  j_0), \,  X_{i_0j_0} \geq  \sqrt{n}.
\end{eqnarray*}
We set $A_j=\{ X_{i_0j} \geq  \sqrt{n} \}$, we have 
$$A \subset \bigcup\limits_{j=1}^{n} A_j.$$

Furthermore,
$$\PP( A_j) =\int_{\sqrt{n} }^{+\infty} e^{-x}dx=e^{-\sqrt{n}}.$$

Therefore, we get
$$\PP(A) \leq \sum_{j=1}^n \PP(A_j) =ne^{-\sqrt{n}}.$$
\end{proof}

\begin{Lem}\label{lem:p(A)H2}
If $H_2(\varepsilon)$  is satisfied and $n \geq 3$ then we have 
$$\PP(A)=0.$$
\end{Lem}
\begin{proof}
We follow the same strategy as before.\\
We have  $\PP(A) \leq \PP\big( \sigma_n(\mathcal{D}) \leq n^{-3/2}\big).$
As $\mathcal{D}$ is a diagonal matrix with coefficients $1/(\sum_{j=1}^n a_{ij}+\varepsilon_{ij})$, we deduce 
\begin{eqnarray*}
\sigma_n(\mathcal{D}) \leq n^{-3/2} 
&\Rightarrow & \exists i_0, \,  \sum_{j=1}^n a_{i_0j}+\varepsilon_{i_0j} \geq  n^{3/2}\\
& \Rightarrow & n \varepsilon \geq  n^{3/2}-1, \quad \textrm{ since } \sum_{j=1}^na_{ij}=1\textrm{ and } |\varepsilon_{ij}|<\varepsilon\\
& \Rightarrow &  \varepsilon \geq  \sqrt{n}-\dfrac{1}{n}.
\end{eqnarray*}
As $\varepsilon<1$ and $n \geq 3$, the last inequality is impossible, thus $\PP\big( \sigma_n(\mathcal{D}) \leq n^{-3/2}\big)=0$ and this proves the desired result.
\end{proof}

In order to to give a bound on $\PP(B)$ we are going to use the following theorem due to Tao and Vu, see \cite{TaoVu}, and see also the erratum in \cite{TaoVuerratum}.

\begin{Thm}\label{thm:tao}
Let $Y$ be a random variable with mean zero and bounded second
moment, and let $\gamma \geq 1/2$, $a \geq 0$ be constants. Then there is a constant $c$ depending
on $Y$, $\gamma$, and $a$ such that the following holds. Let $\mathcal{Y}$ be the random matrix of size $n \times n$ whose
entries are independent and identically distributed copies of $Y$, let $M$ be a deterministic matrix satisfying $\|M\|_2\leq  n^{\gamma}$. Then
$$\PP \big(\sigma_n(M+\mathcal{Y}) \leq n^{-(2a+2)\gamma+1/2} \big) \leq c \Big( n^{-a+o(1)}+\PP\big(\| \mathcal{Y}\|_2 \geq n^{\gamma}\big) \Big).$$
\end{Thm}

In order to bound $\PP(B)$, we need the following lemma.

\begin{Lem}\label{lem:boundproduit}
Let $\mathcal{A}$ and $\mathcal{B}$ be two $n\times n$ matrices, we have
$$\sigma_n(\mathcal{A}) \times \sigma_n(\mathcal{B}) \leq \sigma_n(\mathcal{A}\mathcal{B}).$$
\end{Lem}

\begin{proof}
When $\mathcal{A}$ or $\mathcal{B}$ are singular then this lemma is trivial. Thus, it remains to study the situation when $\mathcal{A}$ and $\mathcal{B}$ are non-singular.\\
When $\mathcal{M}$ is an $n\times n$ matrix we have $\| \mathcal{M}^{-1} \|_2^{-1}=\sigma_n(\mathcal{M})$, see \cite[Theorem 3.3]{Demmel}. Thus
$$\dfrac{1}{\sigma_n(\mathcal{A}\mathcal{B})} =\| (\mathcal{A}\mathcal{B})^{-1} \|_2=\|\mathcal{B}^{-1} \mathcal{A}^{-1} \|_2 \leq \|\mathcal{B}^{-1}\|_2 \times \| \mathcal{A}^{-1} \|_2.$$
Therefore $\sigma_n(\mathcal{A}\mathcal{B}) \geq \|\mathcal{B}^{-1}\|_2^{-1} \times \| \mathcal{A}^{-1} \|_2^{-1}$, which gives the desired result.
\end{proof}

The previous lemma allows to reduce our study to the singular value of $\mathcal{X}$.

\begin{Lem}\label{lem:BdansC}
We set $C=\{ \sigma_n(\mathcal{X}) \leq n^{-b+3/2}\}$. We have $B \subset C$.
\end{Lem}
\begin{proof}
With our notations we have $\mathcal{M}=\mathcal{D}\mathcal{X}$ and by Lemma~\ref{lem:boundproduit} we have
$$(\star) \quad \sigma_n(\mathcal{D}) \times \sigma_n(\mathcal{X}) \leq \sigma_n(\mathcal{M}).$$
If the event $B$ is realized then by definition we have
$$\sigma_n(\mathcal{M}) \leq n^{-b} \textrm{ and } \sigma_n(\mathcal{D})\geq n^{-3/2}.$$
The inequality $(\star)$  implies
$$\sigma_n(\mathcal{X}) \leq n^{-b+3/2}.$$
This gives $B \subset C$.\\
\end{proof}

\begin{Prop}\label{prop:p(B)}
If the hypothesis $H_1$ is satisfied, there exists a constant $c$ such that the following holds 
$$\PP(B) \leq c \Big( n^{-\frac{b-1}{3}+1+o(1)}+n^2e^{-\sqrt{n}-1} \Big).$$
\end{Prop}

\begin{proof}
We are going to apply Theorem~\ref{thm:tao} to $\sigma_n(\mathcal{X})$.\\
We set $Y=X-1$, where $X$ is the exponential distribution with the probability density function $f(x)=e^{-x}$. Then $Y$ is a random variable with mean zero and bounded second moment.\\
Furthermore, we denote by $M$  the $n \times n$ matrix with all its entries equal to $1$. We denote by $\mathcal{Y}=(Y_{ij})$ the $n \times n$ matrix where its coeffcients $Y_{ij}$ are independent and identicaly distributed copies of $Y$.
Therefore, the matrix $M+\mathcal{Y}$ corresponds to our matrix $\mathcal{X}$.\\
Furthermore, we remark easily that we have $\|M\|_2=n$. Then we can set $\gamma=3/2$, and we have $\|M\|_2 \leq n^{\gamma}$.\\
Now, we are going to  bound  $\PP\big( \|\mathcal{Y}\|_2 \geq n^{3/2} \big)$.
We recall the classical bound, see \cite{GolubvanLoan},
$$\| \mathcal{Y} \|_2 \leq n \max_{i,j} |Y_{ij}|,$$
where $\mathcal{Y}=(Y_{ij})$.
Therefore, we have 
$$\|\mathcal{Y}\|_2\geq n^{3/2} \Rightarrow \max_{i,j}|Y_{ij}| \geq \sqrt{n}.$$
We denote by $C_{i,j}$ the following set
$$C_{i,j}=\{|Y_{ij}|\geq \sqrt{n} \}.$$
We deduce then the following inclusion
$$\{ \|\mathcal{Y}\|_2\geq n^{3/2} \} \subset \bigcup\limits_{i,j=1}^n C_{i,j}.$$
By definition of $\mathcal{Y}$, we have $Y_{ij}=X_{ij}-1$ where $X_{ij}$ follows the exponential distribution. Therefore $X_{ij} \geq 0$ and 
$$|Y_{ij}| \geq \sqrt{n} \iff |X_{ij}-1| \geq \sqrt{n} \iff X_{ij} \geq \sqrt{n}+1.$$
This gives
$$\PP(C_{i,j})=\PP\big( |Y_{ij}| \geq \sqrt{n} \big) =\PP(X_{ij} \geq \sqrt{n}+1)=\int_{\sqrt{n}+1}^{+\infty} e^{-x} dx=e^{-\sqrt{n}-1},$$
it follows
$$\PP\big( \|\mathcal{Y}\|_2\geq n^{3/2} \big) \leq \sum_{i,j=1}^n \PP(C_{i,j}) \leq n^2e^{-\sqrt{n}-1}.$$
Then Theorem~\ref{thm:tao} gives
$$\PP\Big(\sigma_n(M+\mathcal{Y}) \leq n^{-3(2a+2)/2+1/2}\Big) \leq c \Big( n^{-a+o(1)}+n^2e^{-\sqrt{n}-1} \Big).$$
By construction, we have $M+\mathcal{Y}=\mathcal{X}$, then if we set 
$$-b+\dfrac{3}{2}=-(2a+2)\times \dfrac{3}{2}+\dfrac{1}{2}$$
 then we have
 $$a=\dfrac{b-1}{3}-1$$
 and 
 $$\PP\big(\sigma_n(\mathcal{X}) \leq n^{-b+3/2}\big) \leq c \Big( n^{-\frac{b-1}{3}+1+o(1)}+n^2e^{-\sqrt{n}-1} \Big).$$
 We obtain the desired bound thanks to Lemma~\ref{lem:BdansC}.
\end{proof}

\begin{Prop}\label{prop:p(B)H2}
If the hypothesis $H_2(\varepsilon)$ is satisfied, there exists a constant $c$ such that the following holds 
$$\PP(B) \leq c \Big( n^{-\frac{b-1}{3}+1+o(1)}\Big).$$
\end{Prop}

\begin{proof}
As before we are going to apply Theorem~\ref{thm:tao} to $\sigma_n(\mathcal{X})$.\\
We set $Y=\varepsilon_{11}$. As  $\varepsilon_{11}$ is a random variable with mean zero and $|\varepsilon_{11}|<\varepsilon$, we deduce that $\varepsilon_{11}$ has a bounded second moment and thus satisfies the hypothesis of Theorem~\ref{thm:tao}.\\
Furthermore, we denote by $M$  the $n \times n$ matrix with  entries equal to $a_{ij}$. We denote by $\mathcal{Y}=(Y_{ij})$ the $n \times n$ matrix where its coeffcients $Y_{ij}$ are independent and identicaly distributed copies of $Y$.
Therefore,  our matrix $\mathcal{X}$ corresponds to the matrix $M+\mathcal{Y}$ in Theorem~\ref{thm:tao}.\\
Furthermore, we  have $\|M\|_2\leq \Big(\sum_{i=1}^n \sum_{j=1}^n a_{ij}^2\Big)^{1/2}\leq n^{1/2}$, since $\sum_{j=1}^n a_{ij}=1$. Then we can set $\gamma=3/2$, and we have $\|M\|_2 \leq n^{\gamma}$.\\
Moreover, we have $\|\mathcal{Y}\|_2\leq \Big(\sum_{i=1}^n \sum_{j=1}^n \varepsilon_{ij}^2 \Big)^{1/2} \leq n\varepsilon < n$, since $|\varepsilon_{ij}|<\varepsilon <1$.\\
Thus it is impossible to have $\|\mathcal{Y}\|_2 \geq n^{\gamma}$ and then  $\PP\big( \|\mathcal{Y}\|_2 \geq n^{\gamma} \big)=0$.\\
Theorem~\ref{thm:tao} gives
$$\PP\Big(\sigma_n(M+\mathcal{Y}) \leq n^{-3(2a+2)/2+1/2}\Big) \leq c \Big( n^{-a+o(1)} \Big).$$
By construction, we have $M+\mathcal{Y}=\mathcal{X}$, then if we set as in Proposition~\ref{prop:p(B)}
$$-b+\dfrac{3}{2}=-(2a+2)\times \dfrac{3}{2}+\dfrac{1}{2}$$
 we have
 $$a=\dfrac{b-1}{3}-1$$
 and 
 $$\PP\big(\sigma_n(\mathcal{X}) \leq n^{-b+3/2}\big) \leq c \Big( n^{-\frac{b-1}{3}+1+o(1)} \Big).$$
 We obtain the desired bound thanks to Lemma~\ref{lem:BdansC}.

\end{proof}

Now, we can prove Proposition~\ref{prop:tao}.\\
Thanks to $(\sharp)$, we have $\PP(\sigma_n(\mathcal{M}) \leq n^{-b})  \leq  \PP(A)+\PP(B)$.\\
Furthermore, by Lemma~\ref{lem:p(A)}, Lemma~\ref{lem:p(A)H2},  Proposition~\ref{prop:p(B)} and  Proposition~\ref{prop:p(B)H2} there exists a constant $c$ such that 
$$\PP(\sigma_n(\mathcal{M}) \leq n^{-b})  \leq  \PP(A)+\PP(B) \leq c \Big( n^{-\frac{b-1}{3}+1+o(1)}+ne^{-\sqrt{n}}+n^2e^{-\sqrt{n}-1} \Big),$$
which gives the desired result.

\section*{Conclusion}
We have shown that, with high probability, the use of an unbounded algorithm can be  efficient. Indeed, in Webb's algorithm we cannot bound the number of queries in term of the number of players, but if we use this algorithm in our probabilistic framework, then  $\PP\big( C(\mu_1,\ldots,\mu_n)\big) \geq n^{12} \big)$ is  small.



 

\end{document}